\documentclass[aps,prl,twocolumn, 10pt, superscriptaddress, nofootinbib]{revtex4-2}
\usepackage{graphicx} 
\usepackage{QI}
\usepackage{hyperref}
\usepackage{natbib}
\usepackage{orcidlink}

\makeatletter
\def\@seccntformat#1{\hspace{0em}}
\def\section{\@startsection{section}{1}{0pt}%
    {-1.0ex plus -.5ex minus -.2ex}%
    {0.5ex plus .2ex}%
    {\normalfont\bfseries}}
\makeatother

\begin{document}
\title{Circuit Knitting for Continuous-Variable Quantum States}
\author{Shao-Hua Hu \orcidlink{0009-0001-9987-5020} }
\email{shhphy@gmail.com}
\affiliation{Department of Physics, National Tsing Hua University, Hsinchu 30013, Taiwan}
\author{Ray-Kuang Lee \orcidlink{0000-0002-7171-7274} }
\email{rklee@ee.nthu.edu.tw}
\affiliation{Department of Physics, National Tsing Hua University, Hsinchu 30013, Taiwan}
\affiliation{Institute of Photonics Technologies, National Tsing Hua University, Hsinchu 30013, Taiwan}
\affiliation{Department of Electrical Engineering, National Tsing Hua University, Hsinchu 30013, Taiwan}

\begin{abstract}
In finite-dimensional systems, circuit knitting can be used to simulate non-classical quantum operations using a limited set of resources.
In this work, we extend circuit knitting techniques to infinite-dimensional quantum systems.
We develop a general theoretical framework for simulating non-Gaussian states from the given set of available states.
Also, we establish fundamental constraints with the no-go theorem on the circuit knitting of multi-mode Gaussian operations, by showing that the exact knitting with separable operations requires infinite sampling overhead.
We further explore several applications of our theory, including simulation of approximate Fock states, GKP state generation, and cat-state amplification.
\end{abstract}
\maketitle

\paragraph{Introduction}
The physical realization of quantum computing is one of the most challenging problems in many scenarios. 
In finite-dimensional quantum computing, circuit knitting \cite{Peng2020, Mitarai2021, Piveteau2024} has been employed to simulate large circuits on smaller quantum processors.
Circuit knitting is interesting for two reasons: One is from the practical point of view, where device requirements can be relaxed.
Second, from the fundamental point of view, circuit knitting provides a method for simulating complex quantum systems from simpler ones.
The core idea behind this is the quasi-probability decomposition (QPD) \cite{Pashayan2015, Mitarai2021}, in which we can decompose a single non-local gate as an affine combination of the local operations.
Negative weights in the decomposition can be simulated via post-processing on the measurement outcome.
Consequently, this approach introduces additional sampling complexity that scales exponentially with the number of non-local gates.
Despite this limitation, circuit knitting remains a powerful tool that enhances our ability to manipulate the quantum system.
Other than simulating non-local gates, similar concepts have also been investigated, like error mitigation \cite{Cai2023, Tsubouchi2023} or virtual distillation \cite{Yuan2024, Takagi2024}.

In continuous-variable (infinite-dimensional) systems, Gaussian states and operations \cite{Weedbrook2012} are widely available in various physical platforms.
However, non-Gaussian states and operations are crucial in various quantum information tasks, like entanglement distillation \cite{Fiurasek2002, Eisert2002, Giedke2002}, computational speed-up \cite{Lloyd1999, Bartlett2002, Mari2012}, and more \cite{Niset2009, Serafini2020}.
Recently, some of the works have already adopted the QPD technique into the continuous variable system, like the error mitigation in bosonic circuits \cite{Teo2025} and increasing the squeezing level of the continuous variable state \cite{Anai2024, Endo2024}.
In this work, we aim to use circuit knitting to prepare a general non-Gaussian state and give several concrete analyses on the QPD and the sampling overhead.
Furthermore, we also study the circuit knitting of the non-local Gaussian operation, and show a no-go theorem on the circuit knitting of the Gaussian unitary.

\bigskip
\paragraph*{Quasi-probability decomposition of continuous variable state}\label{sec::QPD-in-CV}
In this work, we consider the infinite-dimensional Hilbert space $\mathcal{H}$ equipped with the canonical position and momentum operators $\hat{Q}$ and $\hat{P}$, which satisfy the commutation relation $[\hat{Q},\hat{P}] = i\hat{\mathbb{I}}$ (here we set $\hbar = 1$). 
From these, we define the annihilation and creation operators $\hat{a} = \tfrac{1}{\sqrt{2}}(\hat{Q}+i\hat{P}), 
\quad 
\hat{a}^\dagger = \tfrac{1}{\sqrt{2}}(\hat{Q}-i\hat{P})$,
and the number operator $\hat{N} = \hat{a}^\dagger \hat{a}$, whose eigenbasis is the Fock basis $\{\ketv{n}_N : n \in \mathbb{N}_0\}$. 
Furthermore, two types of unitary operators will be used for the latter discussion, namely, the \textbf{displacement operator} $\hat{D}(\alpha) = \exp(\alpha \hat{a}^\dagger - \alpha^* \hat{a})$, and the \textbf{squeezing operator} $\hat{S}(\zeta) = \exp\!\left(\tfrac{\zeta}{2}\hat{a}^{\dagger 2} - \tfrac{\zeta^*}{2}\hat{a}^2\right)$.
We denote by $\mathcal{D}(\mathcal{H})$ the set of density operators on $\mathcal{H}$, i.e., all positive semi-definite $(\hat{\rho}\succeq 0)$ and normalized $(\operatorname{Tr}\hat{\rho} = 1)$ operators. 
For a pure state $\ketv{\psi} \in \mathcal{H}$, we will use the shorthand notation $\hat{\psi} = \proj{\psi} \in \mathcal{D}(\mathcal{H})$.

We formally define QPD as follows:
\begin{definition}\label{def_QPD}
    Let $\mathbb{F}\subset\mathcal{H}$ be the closed and convex set of available states.
    The QPD of state $\hat{\psi}\in\mathcal{H}$ over $\mathbb{F}$ is a set $\{(q_x,\hat{\rho}_x)\in(\mathbb{R},\mathbb{F})\}$, s.t. $\hat{\psi} = \sum_{x\in X}q_x\hat{\rho}_x$ with $\bar{\gamma} = \sum_{x\in X}|q_x|<\infty$.
\end{definition}
In short, the QPD can be treated as the taking ensemble average over some quasi-probability distribution (since $\sum_{x\in X}q_x=1$), where conditions of positive probability are released.
Those negative weights can be reproduced by post-processing.
In particular, a Monte Carlo simulation protocol is given as follows:
First, prepare the state $\hat{\rho}_x$ with the probability $P_x= \frac{|q_x|}{\bar{\gamma}}$.
Then, obtain the measurement outcome $\trace{\hat{O}\hat{\rho}_x}$ for any desired observable $\hat{O}$.
Finally, we re-weight this outcome by the factor $\frac{q_x}{|q_x|}$.
Hence for any observable $\hat{O}$, we have $\sum_{x\in X}\;\frac{q_x}{|q_x|}\trace{\hat{O}\frac{|q_x|}{\bar{\gamma}}\hat{\rho}_x} = \left(\bar{\gamma}\right)^{-1}\trace{\hat{O}\hat{\psi}}$.
So once the QPD of $\hat{\psi}$ is given, we may use it to estimate the outcome statistics of the target state $\hat{\psi}$, with respect to any observable $\hat{O}$.
A carton of this protocol is illustrated in FIG.\ref{fig: intro circuit kintting}.
\begin{figure}[htb]
    \centering
    \includegraphics[width=1.0\linewidth]{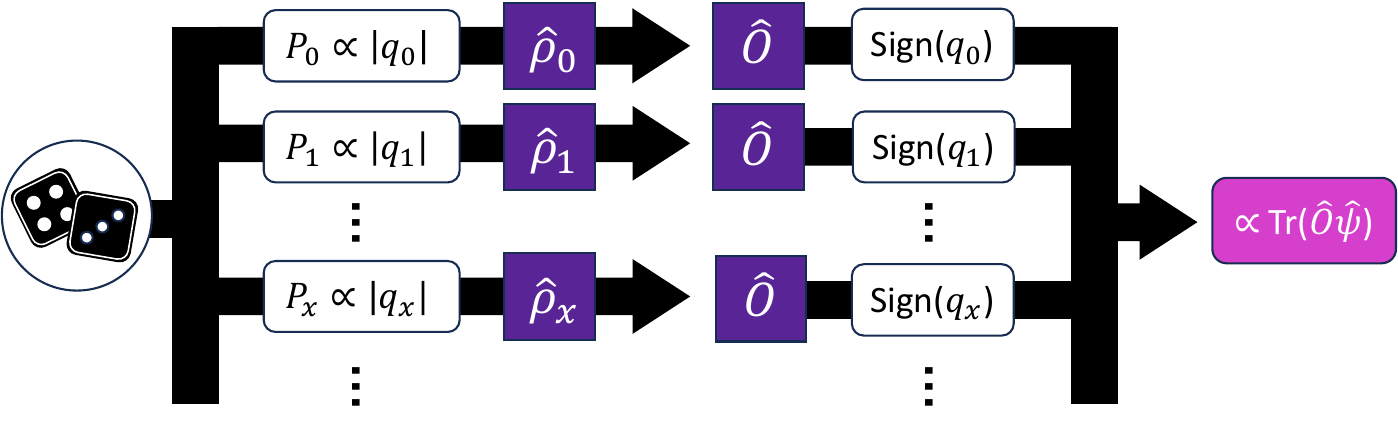}
    \caption{
    This figure illustrates the simulation protocol that allows us to estimate the outcome statistics of the target state $\hat{\psi}$, using only the available states in $\mathbb{F}$ with post-processing.
    }
    \label{fig: intro circuit kintting}
\end{figure}

But there is no free lunch in real life, and this approach does have a consequence.
The additional factor $\bar{\gamma}$ plays the role of sampling overhead.
More precisely, one needs $\bar{\gamma}^2$ times more additional sampling numbers to achieve the same statistical accuracy \cite{Pashayan2015, Peng2020, Mitarai2021}.
Hence, the value $\bar{\gamma}$ should be minimized for the state we are interested in, so one can define the optimal sampling overhead as 
\begin{align}
    \gamma_\mathbb{F}(\hat{\psi}) := \lim_{\epsilon\rightarrow 0^+} \inf_{(q_x,\hat{\rho}_x)} \left\{
\sum_{x\in X} |q_x|
\;\Big|\;
\onorm{\sum_{x\in X}q_x\hat{\rho}_x - \hat{\psi}}{1} < \epsilon
\right\},
\end{align}
were the 1-norm are defined as $\onorm{\hat{A}}{1} = \trace{\sqrt{\hat{A}^\dag\hat{A}}}$.
In the language of resource theory, the optimal sampling overhead is equivalent to the standard (free) robustness \cite{Chitambar2019, Lami2021}, which is a well-defined resource measure.
We remark that the QPD over $\mathbb{F}$ does not generally exist for all $\hat{\psi}$.
So here, when the optimal sampling overhead diverges, it is sure that the QPD does not exist, and vice versa.

\bigskip
\paragraph{Single-mode non-Gaussian state}
In the single-mode scenario, the most available state should be the coherent state $\ketv{ \alpha }:= \hat{D}(\alpha)\ketv{0}_N$.
We denote the classical states $\mathcal{C}_1$ as the closure of the convex hull of the coherent states \cite{Lami2021}.
\begin{theorem}\label{thm_no-go thm for CQPD}
    For any pure state $\hat{\psi}\notin \mathcal{C}_1$, the QPD over $\mathcal{C}_1$ does not exist.
\end{theorem}
\begin{proof}
It has been shown in \cite{Lami2021} that the standard robustness diverges for all pure non-classical states.
\end{proof}
\noindent
Another way of thinking about the theorem~\ref{thm_no-go thm for CQPD} is that, for the non-classical pure state, its P function can't be written into a quasi-distribution, so the QPD over classical states will not exist.
This result leads us to consider the larger sets of the available states, and the choice of the Gaussian states naturally arises.
We denote a pure Gaussian state by the corresponding displacement and squeezing operator, i.e. $\ketv{\alpha,\zeta}:= \hat{D}(\alpha)\hat{S}(\zeta)\ketv{0}_N$.
Then we take the closure of the convex hull of Gaussian states and denote it as $\mathcal{G}_1$.
However, we are unable to construct an explicit QPD for a pure non-Gaussian state over $\mathcal{G}_1$, nor prove the divergence of the standard robustness.
Hence, we make it as the conjecture~\ref{con_nogo-thm-NGQPD}.
\begin{conjecture}\label{con_nogo-thm-NGQPD}
    Let $\hat{\psi}\notin \mathcal{G}_1$ be a pure state, then the QPD over $\mathcal{G}_1$ does not exist.
\end{conjecture}
We notice that conjecture \ref{con_nogo-thm-NGQPD} directly leads to the fact that Hudson's theorem \cite{Hudson1974} does not hold for mixed states \cite{Broecker1995, Chabaud2021, Mitarai2021a}.
\begin{corollary}(\;$\mathcal{G}_1\subsetneq \mathcal{W}_+$\;)
There exists a mixed state with a positive Wigner function, but it can't be written as a convex combination of the Gaussian state.
\end{corollary}
\begin{proof}
    Set $\hat{\rho}_p = (1-p)\proj{1}_N + p \proj{0}_N$, then $\exists p_0 <1$, such that $\hat{\rho}_{p_0}$ has positive Wigner function.
    If Hudson's theorem also holds for $\hat{\rho}_{p_0}$, we have
    $\hat{\rho}_{p_0} = \sum_{x\in X} P_x\hat{g}_x$.
    Thus $\proj{1}_N = \frac{1}{1-p_0}(\sum_{x\in X} P_x\hat{g}_x-p_0\proj{0}_N)$, which contradiction to the conjecture \ref{con_nogo-thm-NGQPD}.
\end{proof}

Now, we extend the available state into the (displaced) cat-states, in which we define $|Cat(\alpha, \beta,\theta)\rangle = \frac{1}{\sqrt{\mathcal{N}(\alpha,\beta,\theta)}}(|\alpha\rangle + e^{i\theta}|\beta\rangle)$ as the (displaced) cat-states, with the normalized constants $\mathcal{N}(\alpha,\beta,\theta) = 2+(e^{i\theta}e^{\alpha^*\beta}+e^{-i\theta}e^{\alpha\beta^*})e^{\frac{-|\alpha|^2-|\beta|^2}{2}}$.
We denoted the closure of the convex hull of the cat states as $\mathcal{C}_2$.
\begin{theorem}\label{thm_c2QPD}
The QPD over $\mathcal{C}_2$ exists for a pure state $|\psi\rangle$, if $\int_{\mathbb{C}}d^2\alpha\;|\langle\alpha|\psi\rangle|<\infty$.
Furthermore, the corresponding QPD that achieves this overhead is given by
\begin{align}
    &|\psi\rangle\langle\psi| =  \frac{1}{4\pi^2}\int_{\mathbb{C}^2}d^2\alpha d^2\beta\;|\langle \alpha|\psi\rangle\langle\psi|\beta\rangle|\notag\\
   &\times\left(\mathcal{N}(\alpha,\beta,\phi)\hat
    {C}_{\alpha,\beta,\phi} -\mathcal{N}(\alpha,\beta,\phi+\pi)\hat{C}_{\alpha,\beta,\phi+\pi}\right)\label{eq_c2_QPD}.
\end{align}
Here, $\hat{C}_{\alpha,\beta,\phi} := \proj{Cat(\alpha,\beta,\phi)}$ and $\langle \alpha|\psi\rangle\langle\psi|\beta\rangle = |\langle \alpha|\psi\rangle\langle\psi|\beta\rangle|e^{i\phi}$.
\end{theorem}
\begin{proof}
By expanding the RHS of Eq.~\eqref{eq_c2_QPD}, one can check that this statement is true.
The detailed calculation can be found in Appendix A.
\end{proof}
Notice that for the states that can be written into a finite superposition of Fock states or Gaussian states, it must satisfy the condition $\int_{\mathbb{C}}d^2\alpha\;|\langle\alpha|\psi\rangle|<\infty$.
Based on this observation, we can already simulate a large class of non-Gaussian states with cat states.
Although this overhead is finite, it scales very badly in general.
Hence, we further extend the available states to include arbitrary superpositions of two Gaussian states, i.e. $|\psi\rangle = c_1|g_1\rangle+c_2|g_2\rangle$, and denote the closure of the convex hull as $\mathcal{G}_2$.
\begin{theorem}\label{thm_g2QPD}
The QPD over $\mathcal{G}_2$ exists for pure state $|\psi\rangle$, if there exist a expansion $\ketv{\psi} = \sum_{x} c_x|g_x\rangle$ s.t. $\sum_{x}|c_x|<\infty$.
Furthermore, the corresponding QPD that achieves this overhead is given by
\begin{align}
   &\proj{\psi}=\sum_{x=1}^{R_G}|c_x|^2\proj{g_x}\notag\\
   &+\sum_{x>x'}|c_{x}c_{x'}^*|\left(\frac{\mathcal{N}^+_{x,x'}}{2}\proj{g_{x,x'}^+} - \frac{\mathcal{N}^-_{x,x'}}{2}\proj{g_{x,x'}^-}\right).\label{eq_g2_QPD}
\end{align}
Here, $|c_xc_{x'}^*|e^{\theta_{x,x'}}:= c_xc_{x'}^*$ and $|g_{x,x'}^\pm\rangle := \frac{1}{\sqrt{\mathcal{N}^\pm_{x,x'}}}(|g_x\rangle \pm e^{i\theta_{x,x'}}|g_{x'}\rangle)$.
\end{theorem}
\begin{proof}
By expanding the RHS of Eq.~\eqref{eq_g2_QPD}, one can check this statement is true, for the same reason as Theorem \ref{thm_c2QPD}.
The detailed calculation can be found in Appendix B.
\end{proof}
We notice that the sampling overhead of the QPD in Theorem \ref{thm_g2QPD} has the same form as the Gaussian extent defined in \cite{Hahn2024}, which is a valid resource measure of non-Gaussianity.

\bigskip
One of the application on the Theorem \ref{thm_g2QPD} is the simulation of Gottesman-Kitaev-Preskill (GKP) states \cite{Gottesman2001}, defined as 
\begin{align}
    |\mu\rangle_{GKP} = \sum_{n=-\infty}^\infty |(2n+\mu)\sqrt{\pi}\rangle_Q,
\end{align}
where $\mu = \{0,1\}$ and $\ketv{q}_Q$ is the position eigenstate.
The GKP state has significant applications in full-tolerance quantum computing.
However, the ideal GKP state is not a physical state, since it can't be normalized.
Therefore, we need to approximate the ideal GKP state.
A straightforward approximated GKP state is given in the form
\begin{align}
    |\Tilde{\mu}\rangle_{GKP} = \sum_{n=-\infty}^\infty c_n\hat{D}_{2n+\mu}\hat{S}(r)|0\rangle_N.
\end{align}
Here we set $\hat{D}_{m} = \hat{D}(m\sqrt{\frac{\pi}{2}})$ and $c_n$ is a positive coefficient, which can have different choices and usually depends on the state generation protocol.
However, as long as we have access to $\mathcal{G}_2$ state, we may construct it using Theorem ~\ref{thm_g2QPD}.
Since it is already written in the superposition of the Gaussian states, the corresponding sampling overhead is then equal to $(\sum_n c_n)^2$.

For the random walk approach of GKP state generation \cite{Lin2020, Wu2025}, the approximated GKP zero state has the form
\begin{align}
    |\Tilde{0}\rangle_{GKP} = \frac{1}{\sqrt{\mathcal{N}}}\sum_{n=0}^L \binom{L}{n} \hat{D}_{2n-L}\hat{S}(r)|0\rangle_N
\end{align}
where $L$ is the number of steps in the random walk.
With circuit knitting we obtain the sampling overhead $\bar{\gamma} = \frac{4^{L}}{\mathcal{N}}$, hence by the Hoeffding’s inequality, $\frac{16^L}{\mathcal{N}^2}$ additional shot are needed.
But compared to the random walk approach, we need to implement the post-selection on the auxiliary qubit with $\frac{1}{2}$ success probability for each round.
So $2^L$ additional shots are needed.
However, the normalization constant is $\mathcal{N}\approx\sum_{n=0}^L \binom{L}{n}^2 = \binom{2L}{L}$, when $\langle \zeta = r|\hat{D}^\dag_{2n}\hat{D}_{2m}|\zeta = r\rangle = e^{-e^{2r}\pi\frac{|n-m|^2}{2}} \approx \delta_{nm}$.
Thus, the ratio of the additional shots between the two approaches is given by
\begin{align}
    \frac{S_{CK}}{S_{RW}} = \frac{8^L}{\mathcal{N}^2}\approx 8^L\frac{(L!)^4}{(2L!)^2}\approx \pi L 2^{-L},
\end{align}
We obtain the last approximation by Stirling's formula.
This shows that additional shots of circuit knitting are exponentially smaller than the random walk approach.
We plot the number of additional shots in the FIG.\ref{fig:shots-GKP}
\begin{figure}
    \centering
    \includegraphics[width=1.0\linewidth]{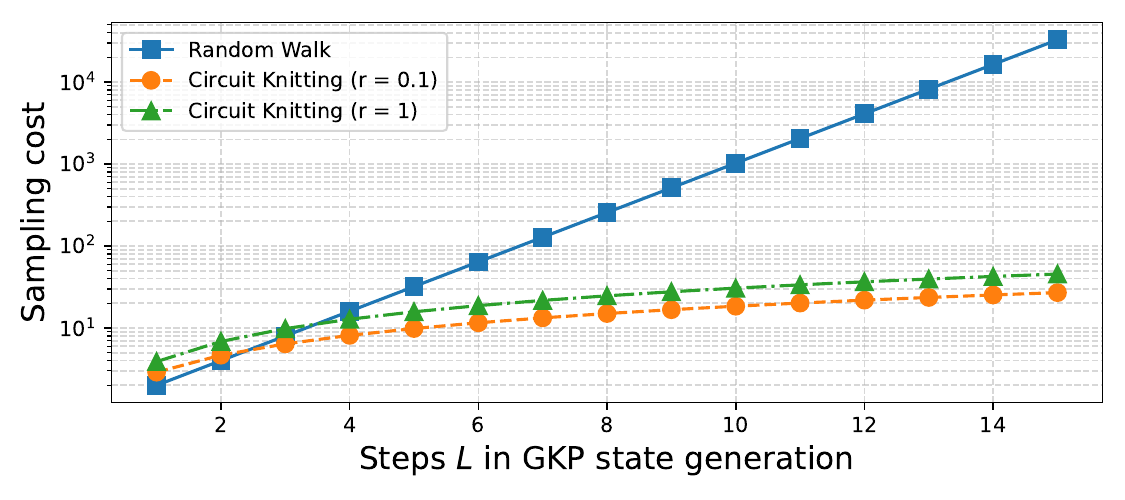}
    \caption{
    We plot the additional sampling cost for different approaches to GKP state generation. Including random walk (blue square), circuit knitting under $r = 0.1$ (orange circle), and circuit knitting under $r = 1$ (green triangle)
    }
    \label{fig:shots-GKP}
\end{figure}
Also, another advantage of circuit knitting is the steps of state preparation, where the circuit complexity is independent of the L, unlike the random walk approach.

A big disadvantage of circuit knitting is that one needs to be able to generate large-sized squeezed cat $\propto|\alpha, r\rangle\pm |-\alpha,r\rangle$, since the size of $\alpha$ will give the limit of $L$.
But the random walk approach won't face this difficulty, since the walking distance is fixed for each round.
Also, circuit knitting only allows us to obtain the measurement outcome of the state, unlike the random walk approach.

In short, both methods will be useful in different scenarios.
One can also combine those two methods, where we first generate a GKP state with circuit knitting and then increase the fidelity of the GKP state through a random walk.
\bigskip

\paragraph{Multi-mode Gaussian unitary}
In the multi-mode scenario, we may include the entanglement as the resource.
It has been shown that, for the bipartite pure state with Schmidt decomposition $\ketv{\Psi}=\sum_{n}s_n\ketv{a_n,b_n}$, we have $\gamma_{\mathrm{SEP}}(\hat{\Psi})= 2(\sum_{n}s_n)^2-1$ \cite{Vidal1999, Lami2021}.
Thus, the optimal sampling overhead and decomposition are found.
This motivates us to investigate the circuit knitting of the multi-mode Gaussian unitary.
For any Gaussian unitary $\hat{U}_S$, there exists a corresponding symplectic matrix $S$. 
In which the unitary transformation of quadrature operators $\hat{\bm{R}} = [\hat{Q}_1,\hat{P}_1,\dots,\hat{Q}_M,\hat{P}_M]^T$ can be written as $\hat{U}_S \hat{\bm{R}} \hat{U}_S^\dag = S\hat{\bm{R}}$.
Due to the complected facture of the local operation and classical communication, we take all separable operations as the available set, that is, a quantum channel who's Klaus operator is given in the product form, i.e., $\hat{K}_x = \hat{K}^{(a)}_x\otimes \hat{K}^{(b)}_x$.
We denote the set of separable states or operations as $\mathrm{SEP}$, which should be clear whether it refers to a state or operation in the latter discussion.

\begin{theorem}\label{thm_no-go thm for GCC}
For any non-local Gaussian unitary, the QPD over \textbf{SEP} does not exist.
\end{theorem}
The proof can be found in Appendix C.
Theorem \ref{thm_no-go thm for GCC} indicates the difference between finite and infinite-dimensional systems. As a corollary, we can conclude that techniques like wire cutting \cite{Brenner2023} do not exist in the continuous variable system.
\begin{corollary}\label{cor_no-go thm for GWC}
For a perfect teleportation of a continuous variable state, the QPD over \textbf{SEP} does not exist.
\end{corollary}
\begin{proof}
If the QPD over \textbf{SEP} exists for a perfect teleportation of a continuous variable state with the sampling overhead $\bar{\gamma}$.
Then with two rounds of state teleportation, any non-local unitary can be implemented with sampling overhead $\bar{\gamma}^2$, which contradicts Theorem \ref{thm_no-go thm for GCC}, hence such a QPD does not exist.
\end{proof}
As a remark, both Theorem \ref{thm_no-go thm for GCC} and Corollary \ref{cor_no-go thm for GWC} consider the perfect implementation of non-local operation.
So this does not contradict the result in \cite{Anai2024}, since they are working on the circuit knitting with approximation.

\bigskip
\paragraph{Approximated quasi-probability decomposition}
In reality, there are always some non-vanishing errors, so it is useful to consider the circuit knitting with approximation \cite{Piveteau2022}.
Specialty, for the infinite-dimensional system, we have already shown several no-go theorems, which may be released by introducing a non-vanishing error in the decomposition.
Formally, we define approximate QPD as follows:
\begin{definition}($\epsilon$-QPD)\label{def_approx-QPD}
A state $\hat{\psi}\in\mathcal{D}(\mathcal{H})$ is said to admit an $\epsilon$-QPD over $\mathbb{F}\subset \mathcal{D}(\mathcal{H})$.
If $\forall \epsilon>0$ there exists a corresponding $\hat{\psi}_\epsilon\in\mathcal{D}(\mathcal{H})$, such that the QPD of $\hat{\psi}_\epsilon$ over $\mathbb{F}$ exists and $\onorm{\hat{\psi}-\hat{\psi}_\epsilon}{1}<\epsilon$.
\end{definition}

By allowing the $\epsilon$ difference, one may avoid the divergence in the sampling overhead.
We now show several examples of $\epsilon$-QPD.

\begin{example}[Single-photon state]
A QPD of the approximated single photon state can be given by
\begin{align}
    \hat{\psi}_\epsilon = \frac{1}{1-e^{-\epsilon}}\left(\int_0^{2\pi} \frac{d\theta}{2\pi}\proj{\sqrt{\epsilon} e^{i\theta}} - e^{-\epsilon}\proj{0}_N\right)
\end{align}
One can verify it by checking that the trace norm with a single-photon state is $\onorm{\proj{1}_N-\hat{\psi}_\epsilon}{1} =2(\frac{e^\epsilon-1-\epsilon}{e^{\epsilon}-1})<\epsilon$.
The corresponding sampling overhead is $\bar{\gamma} =\frac{e^\epsilon+1}{e^\epsilon-1}$.
\end{example}
\noindent
A similar decomposition is given by \cite{Bourassa2021}, which achieves the same purpose by using the thermal state minus the vacuum.
However, using the dephasing coherent state can provide a lower overhead in our setting.

Moreover, we can conclude the following corollary:
\begin{corollary}
The $\epsilon$-QPD over $\mathcal{C}_1$ exist for any Fock state $\proj{n}_N$.
\end{corollary}
\begin{proof}
Notice that we already have the example showing that an approximate QPD exists for a single-photon state.
Then, by using techniques like photon bunching, we may use it to construct any Fock state with finite overhead.
More details can be found in Appendix D. 
\end{proof}

\bigskip
\paragraph{Cat-state amplification}
In this section, we introduce a cat-state amplification protocol, which can merge multiple small cat-states into one big cat-state through circuit knitting.

Rather than starting from the multiple cat-states as input, let's consider the state
\begin{align}
    |\Phi(\alpha,\theta)\rangle \propto (|\alpha,\alpha\rangle + e^{i\theta}|-\alpha,-\alpha\rangle) 
\end{align}
By adding the balance beam splitter, we have
\begin{align}
    \hat{U}_B|\Phi(\alpha,\theta)\rangle = |0\rangle_N\otimes |Cat(\sqrt{2}\alpha,-\sqrt{2}\alpha,\theta)\rangle\label{eq_cat-distilation} 
\end{align}
Thus, we obtain the cat-state by tracing out the vacuum, whose size is increased by the factor of $\sqrt{2}$.
However, the state $|\Phi(\alpha,\theta)\rangle$ can be simulated by the cat-states through the circuit knitting, by using a similar construction of the Bell state.
With ability of preparing $|\psi_n(\alpha)\rangle \propto (|\alpha\rangle+e^{i\frac{\theta+n\pi}{2}}|-\alpha\rangle)$, with $n\in \mathbb{Z}_4$ and fix $\alpha$.
We can use it to construct the QPD of $|\Phi(\alpha,\theta)\rangle$, with sampling overhead
\begin{align}
    \bar{\gamma}(\alpha,\theta) = \frac{3}{1+\cos{\theta}e^{-4|\alpha|^2}}.
\end{align}
The details construction of this QPD can be found in Appendix E.

Notice that we can, in principle, repeat this process multiple times by feeding the outcomes of the amplified cat into the next round.
For each round, the size of the cat will be increased by the factor $\sqrt{2}$.
And the recursion formula for the sampling overhead is 
\begin{align}
    \bar{\gamma}(\sqrt{2}\alpha,\theta) \approx 1+2\bar{\gamma}^2(\alpha,\theta),
\end{align}
when $e^{-4|\alpha|^2}\approx 0$ for large enough  $\alpha$.
Although this cat amplification protocol may seem elegant, its super-exponential scaling makes it nearly unusable for practical purposes beyond $L>3$.
We plot the sampling overhead of cat state amplification in FIG.\ref{fig:fack-cat}
\begin{figure}[ht]
    \centering
    \includegraphics[width=1.0\linewidth]{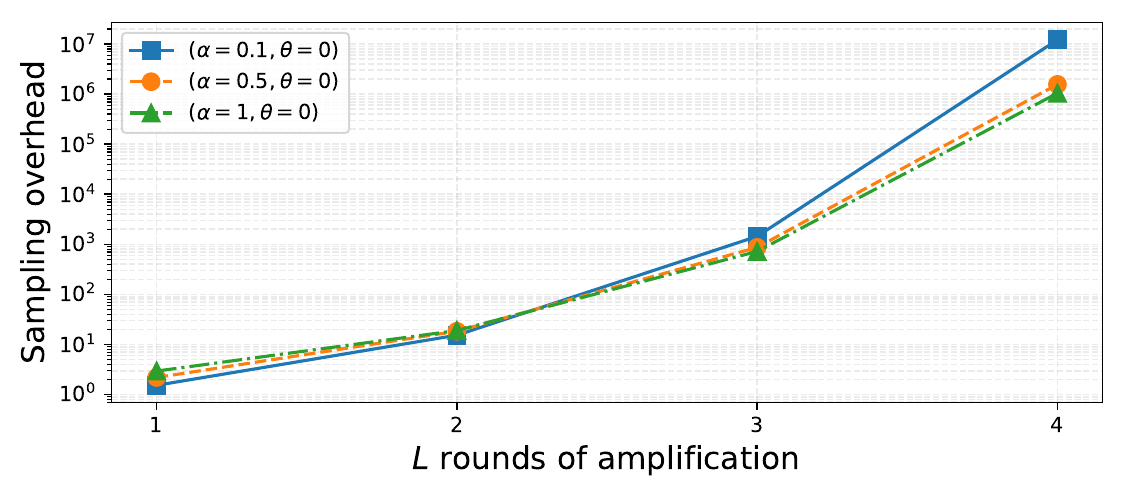}
    \caption{We plot the sampling overhead of cat state amplification, for $\alpha = 0.1$ (blue square), $\alpha = 0.5$ (orange circle), and $\alpha = 1$ (green triangle). 
    }
    \label{fig:fack-cat}
\end{figure}

\bigskip
\paragraph{Conclusion}\label{sec::Conclusion}
In summary, we employ the circuit knitting techniques for the continuous-variable system.
When we only have access to the coherent state, the Fock state can still be simulated approximately.
But when we have access to the cat-state or superposition of two Gaussian states, most of the non-Gaussian states can be exactly simulated with a finite sampling overhead.
For the multi-mode system, we show that any non-local Gaussian unitary can not be simulated with a separable operation.

Several interesting applications are given, including the generation of the approximated GKP state.
Circuit knitting can reduce additional shots exponentially, compared to the random walk approach.
Also, we propose a virtual amplification protocol for the cat state, which allows us to amplify the size of the cat by $\sqrt{2}$, with overhead $\approx 3$.

This study provides both fundamental limitations on circuit knitting for the continuous variable state and also concrete examples of applying circuit knitting to wild range of applications.
As a remark, we want to emphasize that, based on our result, circuit knitting is not efficient compared to the existing classical simulation approach \cite{Bourassa2021, Hahn2024}.
However, a big advantage of the circuit knitting is that it is compatible with the real experiment.
That is, for the real physical implementation of the quantum system, we can freely replace part of the element with the virtual state or operation through the circuit knitting.

\bibliographystyle{apsrev4-2} 
\bibliography{NGQPD}

\clearpage
\onecolumngrid
\begin{center}
\textbf{\large Appendix}
\end{center}

\subsection{A. Quasi-probability decomposition over $\mathcal{C}_2$} \label{appendix_c2}
Let $|\psi\rangle $ be the target pure state.
By subtracting the resolution of identity, we have
\begin{align}
   \proj{\psi} &= \frac{1}{\pi^2}\int_{\mathbb{C}^2}d^2\alpha d^2\beta \;\langle \alpha|\psi\rangle\langle\psi|\beta\rangle(|\alpha\rangle\langle\beta|) \notag\\
    &= \frac{1}{2\pi^2}\int_{\mathbb{C}^2}d^2\alpha d^2\beta\;|\langle \alpha|\psi\rangle\langle\psi|\beta\rangle|(e^{i\phi}|\alpha\rangle\langle\beta|+e^{-i\phi}|\beta\rangle\langle\alpha|) 
\end{align}
Here $e^{i\phi }|\langle \alpha |\psi\rangle\langle \psi|\beta \rangle | = \langle \alpha |\psi\rangle\langle \psi|\beta \rangle$, and we further set $\hat{C}_{\alpha, \beta,\phi}:= |Cat(\alpha, \beta,\phi)\rangle\langle Cat(\alpha, \beta,\phi)|$.
Then it has
\begin{align}
    \mathcal{N}(\alpha,\beta,\phi)\hat{C}_{\alpha,\beta,\phi}
    -\mathcal{N}(\alpha,\beta,\phi+\pi)\hat{C}_{\alpha,\beta,\phi+\pi}
    = 2(e^{i\phi}|\alpha\rangle\langle\beta|+e^{-i\phi}|\beta\rangle\langle\alpha|). 
\end{align}
Hence we obtain the decomposition over $\mathcal{C}_2$, by put it back to the equation, in which
\begin{align}
   |\psi\rangle\langle\psi| &=  \frac{1}{4\pi^2}\int_{\mathbb{C}^2}d^2\alpha d^2\beta\;|\langle \alpha|\psi\rangle\langle\psi|\beta\rangle|\notag\left(\mathcal{N}(\alpha,\beta,\phi)\hat
    {C}_{\alpha,\beta,\phi} -\mathcal{N}(\alpha,\beta,\phi+\pi)\hat{C}_{\alpha,\beta,\phi+\pi}\right). 
\end{align}
The sampling overhead for this decomposition is then 
\begin{align}
    \Bar{\gamma}(\hat{\psi}) = &\frac{1}{4\pi^2}\int_{\mathbb{C}^2}d^2\alpha d^2\beta\;|\langle \alpha|\psi\rangle\langle \psi|\beta\rangle|\left(\mathcal{N}(\alpha,\beta,\phi)+\mathcal{N}(\alpha,\beta,\phi+\pi)\right)
    = \left(\frac{1}{\pi}\int_{\mathbb{C}}d^2\alpha\;|\langle\alpha|\psi\rangle|\right)^2.
\end{align}
Thus it is a valid QPD, when $\int_{\mathbb{C}}d^2\alpha\;|\langle\alpha|\psi\rangle|<\infty$.

\subsection{B. Quasi-probability decomposition over $\mathcal{G}_2$}\label{appendix_g2}
Suppose the target state can be write as $|\psi\rangle = \sum_{x=1}^{R_G} c_x|g_x\rangle$.
With $|g_x\rangle = \ketv{\alpha_x,\zeta_x}$ be a Gaussian state.
Then its density matrix is
\begin{align}
\proj{\psi} =& \sum_{x,x'=1}^{R_G}c_{x}c_{x'}^*\ketv{g_x}\brav{g_{x'}}
=\sum_{x=1}^{R_G}|c_x|^2\proj{g_x}+\sum_{x\neq x'}|c_xc_{x'}^*|e^{i\theta_{x,x'}}\ketv{g_x}\brav{g_{x'}} 
\end{align}
where $|c_xc_{x'}^*|e^{\theta_{x,x'}}:= c_xc_{x'}^*$.
For the non-diagonal parts ($x\neq x'$), we can rewrite them as
\begin{align}
   \sum_{x>x'}|c_xc_{x'}^*|\left(e^{i\theta_{x,x'}}|g_x\rangle\langle g_{x'}|+e^{-i\theta_{x,x'}}|g_{x'}\rangle\langle g_x| \right) 
\end{align}
Defined the state in $\mathcal{G}_2$ as $|g_{x,x'}^\pm\rangle = \frac{1}{\sqrt{\mathcal{N}^\pm_{x,x'}}}(|g_x\rangle \pm e^{i\theta_{x,x'}}|g_{x'}\rangle)$ with normalization constant $\mathcal{N}^\pm_{x,x'} = 2\pm(e^{\theta_{x,x'}}\langle g_x|g_{x'}\rangle +e^{-\theta_{x,x'}}\langle g_{x'}|g_x\rangle)$, hence we have
\begin{align}
     {\mathcal{N}^+_{x,x'}}|g_{x,x'}^+\rangle\langle g_{x,x'}^+| - {\mathcal{N}^-_{x,x'}}|g_{x,x'}^-\rangle\langle g_{x,x'}^-| 
     =2e^{i\theta_{x,x'}}|g_x\rangle\langle g_{x'}|+2e^{-i\theta_{x,x'}}|g_{x'}\rangle\langle g_x|. 
\end{align}
Thus, by subtracting it back from the non-diagonal terms, we obtain the decomposition over $\mathcal{G}_2$ as
\begin{align}
   &\proj{\psi}=\sum_{x=1}^{R_G}|c_x|^2\proj{g_x}+\sum_{x>x'}|c_{x}c_{x'}^*|\left(\frac{\mathcal{N}^+_{x,x'}}{2}\proj{g_{x,x'}^+} - \frac{\mathcal{N}^-_{x,x'}}{2}\proj{g_{x,x'}^-}\right).
\end{align}
The corresponding sampling overhead is 
\begin{align}
    \Bar{\gamma}(\hat{\psi}) =& \sum_{x=1}^{R_G}|c_x|^2+\sum_{x>x'}|c_xc_{x'}^*|\left(\frac{\mathcal{N}^+_{x,x'}+\mathcal{N}^-_{x,x'}}{2}\right)
    =\left(\sum_{x=1}^{R_G}|c_x|\right)^2.
\end{align}
However, its expansion is not unique, so we may take the infimum of the sampling overhead over all possible expansions.
So we obtain the valid QPD, if $\sum_{x=1}^{R_G}|c_x|$ diverge for all possible expansion.

\subsection{C. No-go theorem on the circuit knitting of the Gaussian unitary}\label{appendix_nogothm}
Before proving the theorem, we first introduce Lemma \ref{lemma::LB-gamma_G}. This establishes a relation between the optimal sampling overhead and the covariance matrix for a bipartite pure Gaussian state.
\begin{lemma}\label{lemma::LB-gamma_G}
Let $\ketv{\Psi_{ab}}$ be a bi-pirate pure Gaussian state, then we have
\begin{align}
    \gamma_{\mathrm{SEP}}(\ketv{\Psi_{ab}})\geq 2\sqrt{\mathrm{det}(\Lambda_a)}-1.
\end{align}
In which $\Lambda_{a}$ is the covariance matrix of the subsystem a.
\end{lemma}
\begin{proof}
First, we know that for the pure bipartite state with the Schmidt coefficients $s_i$, we have
\begin{align}
  \frac{1+\gamma_{\mathrm{SEP}}(\ketv{\Psi_{ab}})}{2} = \pnorm{s}{1}^2=(\sum_i s_i)^2.
\end{align}
And the purity of the reduced state is  
\begin{align}
    \trace{\rho_A^2} = \pnorm{s}{4}^4 = \sum_{i}s_i^4.
\end{align}
Then, by adopting Hölder's inequality, it follows that
\begin{align}
    1= \pnorm{s}{2} &\leq \pnorm{s}{4}^4\pnorm{s}{1}^2 = \trace{\rho_A^2}(\frac{1+\gamma_{\mathrm{SEP}}(\ketv{\Psi_{ab}})}{2}) 
    \Rightarrow \gamma_{\mathrm{SEP}}(\ketv{\Psi_{ab}}) \geq \frac{2}{\trace{\rho_A^2}} - 1.
\end{align}
For the Gaussian state, the purity is given by
\begin{align}
    \frac{1}{\trace{\rho_A^2}} = \sqrt{\mathrm{det}(\Lambda_a)}.
\end{align}
Where the covariance matrix is defined as
\begin{align}
    \Lambda_{i,j} := \brav{\Psi_{ab}}\{\hat{R}_i- \langle\hat{R}_i\rangle,\hat{R}_j- \langle\hat{R}_j\rangle\}\ketv{\Psi_{ab}}\text{ with }\; \langle\hat{R}_i\rangle:= \brav{\Psi_{ab}} \hat{R}_i\ketv{\Psi_{ab}}.
\end{align}
So for the bi-partite Gaussian state, we have 
\begin{align}
    \gamma_{\mathrm{SEP}}(\ketv{\Psi_{ab}})\geq 2\sqrt{\mathrm{det}(\Lambda_a)}-1.
\end{align}
\end{proof}
Now we are ready to prove the theorem \ref{thm_no-go thm for GCC}.
Suppose the QPD over $\mathrm{SEP}$ exists for $\hat{U}_S$, which acts on $M_{a}+M_b$ modes.
Then for any $2M_a$ modes state $\ketv{\psi_a}$ and $M_b$ mode state $\ketv{\psi_b}$, it is $\gamma_{\mathrm{SEP}}\left( \hat{U}_S\right)\geq\gamma_{\mathrm{SEP}}\left((\hat{\mathbb{I}}\otimes \hat{U}_S)\ketv{\psi_a,\psi_b}\right)$.
So the optimal sampling overhead for $\hat{U}_S$ is lower bounded by 
\begin{align}
    \gamma_\mathrm{SEP}(\hat{U}_S)\geq \sup_{|\psi_a\rangle,|\psi_b\rangle}\left\{\gamma_{\mathrm{SEP}}\left((\hat{\mathbb{I}}\otimes \hat{U}_S)\ketv{\psi_a,\psi_b}\right)\right\} .
\end{align}
To prove this theorem, we aim to show that this lower bound diverges for all $S\neq S_a\oplus S_b$.

Lets set the two-mode squeezed state as $\ketv{\Phi(r)}=\hat{U}_{B}\ketv{\zeta = r,\zeta = -r}$.
Then we take $\ketv{\Psi_a(r)}$ contain $M_a$ copies of $\ketv{\Phi(r)}$, in which mode m will entangled with mode $m+M_a$, for all $0\leq m< M_a$.
So the covariance matrix of $\ketv{\Psi_a(r)}$ are
\begin{align}
    &\Lambda_{\Psi_a(r)}
    =\begin{pmatrix}
        \cosh 2r&0&-\sinh{2r}&0\\
        0&\cosh{2r}&0&\sinh{2r}\\
        -\sinh{2r}&0&\cosh 2r&0\\
        0&\sinh{2r}&0&\cosh{2r}
    \end{pmatrix}\otimes I_{M_a}. 
\end{align}
Where $I_d$ denotes the d by d identity matrix.
For the subsystem b, we take $\ketv{\Psi_b}=\ketv{0}_N^{\otimes M_b}$.
Then the covalence matrix of $\ketv{\Psi_b}$ is $\Lambda_{\Psi_b} =I_{2M_b}$.

To calculate the optimal sampling overhead, we need the covariance matrix of subsystem b, after a symplectic transform $S$.
We first take the partial trace on the first $M_a$ mode of $\ketv{\Psi_a(r)}$,
so the reduced state $\hat{\rho}_a(r)$ has the covariance matrix
\begin{align}
    \Lambda_{\hat{\rho}_a(r)} = {\cosh{2r}}\mathbb{I}_{2M_a}.
\end{align}
Then by taking $S$ into its block form, i.e.,
\begin{align}
    S = \begin{pmatrix}
        S_{00}&S_{01}\\ S_{10}&S_{11}
    \end{pmatrix} 
\end{align}
The covariance matrix after the symplectic transform $S$ is
\begin{align}
&\Lambda_{\hat{U}_S(\hat{\rho}_a(r)\otimes\proj{\Psi_b})\hat{U}_S^\dag}
={\cosh{2r}}\begin{pmatrix}
      S_{00}S^T_{00} & S_{00}S^T_{10}\\
      S_{10}S^T_{00}&S_{10}S^T_{10}
  \end{pmatrix}+\begin{pmatrix}
      S_{01}S^T_{01} & S_{01}S^T_{11}\\
      S_{11}S^T_{01}& S_{11}S^T_{11}
  \end{pmatrix} \\
\Rightarrow 
&\Lambda_{\mathrm{Tr}_A(\hat{U}_S(\hat{\rho}_a(r)\otimes\proj{\Psi_b})\hat{U}_S^\dag)}
=\cosh{2r}S_{01}S^T_{01} + S_{11}S^T_{11}
\end{align}
So applying Lemma \ref{lemma::LB-gamma_G}, we have 
\begin{align}
    \gamma_\mathrm{SEP}(\hat{U}) 
    &\geq \sup_{|\psi_a\rangle,|\psi_b\rangle} \left\{ \gamma_{\mathrm{SEP}}\left((\hat{\mathbb{I}}\otimes \hat{U}_S)\ketv{\psi_a,\psi_b}\right) \right\} \notag\\
    &\geq \sup_{r\in\mathbb{R}}\left\{ \gamma_{\mathrm{SEP}} \left( (\hat{\mathbb{I}}\otimes \hat{U}_S)\ketv{\Psi_a(r), \Psi_b} \right) \right\} \notag\\
    &\geq 2\sup_{r\in\mathbb{R}} \left\{ \sqrt{\mathrm{det}(\cosh{2r}S_{01}S^T_{01} + S_{11}S^T_{11})} \right\} - 1.
\end{align}
This means, $S_{01}\neq0\Rightarrow \gamma_\mathrm{SEP}(\hat{U}_S)=\infty$.
Also, to make $S$ a valid simplistic transform, we have $S_{01}=0\Leftrightarrow S_{10}=0$.
So $\gamma_\mathrm{SEP}(\hat{U}_S)<\infty$ if and only if $S = S_a\oplus S_b$.
This allows us to conclude that QPD over $\mathrm{SEP}$ does not exist for any non-local Gaussian unitary. 

\subsection{D. Approximate QPD for the Fock state}\label{appendix_Fock}
Recalled that we have the approximate QPD for the single photon state as
\begin{align}
    \hat{\psi}_\epsilon = \frac{1}{1-e^{-\epsilon}}\left(\int_0^{2\pi} \frac{d\theta}{2\pi}\proj{\sqrt{\epsilon} e^{i\theta}} - e^{-\epsilon}\proj{0}_N\right).
\end{align}
Then with the k-copy of $\hat{\psi}_\epsilon$, we can apply the photon bunching as follows:
First, apply the quantum Fourier transform with the beam-splitter, that is, the unitary $\hat{F}_{k}$, which satisfies
\begin{align}
    \hat{F}_{k}\hat{a}^\dag_n\hat{F}_{k}^\dag = \frac{1}{\sqrt{k}}\sum_{m=1}^N e^{\frac{2\pi i nm}{k}} \hat{a}^\dag_m.
\end{align}
After that, we post-select on the other $k-1$ auxiliary modes with no-photon output.
This process can be described by a completed positive trace-non-increasing map as
\begin{align}
    \Tilde{B}_{k}(\hat{\rho}) = \mathrm{Tr}_{Aux}\left(\hat{\mathbb{I}}\otimes \proj{0}_N^{\otimes k-1}\hat{F}_k \hat{\rho} \hat{F}_k^\dag \right).
\end{align}
With the ideal single photon state, it is 
\begin{align}
    \Tilde{B}_{k}(\proj{1}_N^{\otimes k}) = \frac{k!}{k^k}\proj{k}_N.
\end{align}
Hence, we define the approximate K-photon state as
\begin{align}
    \hat{\psi}^{(k)}_\delta :=  \frac{\Tilde{B}_{k}(\hat{\psi}_{\delta}^{\otimes k})}{\mathrm{Tr}\left(\Tilde{B}_{k}(\hat{\psi}_{\delta}^{\otimes k})\right)}
\end{align}
Moreover, the success probability should be bounded by
\begin{align}
    \mathrm{Tr}\left(\Tilde{B}_{k}(\hat{\psi}_{\delta}^{\otimes k})\right)> \frac{\delta^ke^{-k\delta}}{(1-e^{-\delta})^k}\frac{k!}{k^k}
\end{align}
So we can conclude that, 
\begin{align}
    \onorm{\hat{\psi}^{(k)}_\delta - \proj{k}_N }{1} 
    \leq \frac{3(k-ke^{-\delta})^k}{2k!\delta^ke^{-k\delta}}\onorm{\hat{\psi}^{\otimes k}_\delta - \proj{1}_N^{\otimes k} }{1} 
    \leq \frac{3(k-ke^{-\delta})^k}{2k!\delta^ke^{-k\delta}}k\delta := \epsilon_k(\delta).
\end{align}
This is a direct construction on the approximate QPD k-photon state.
Notice that, although we adopt the auxiliary mode in the construction of QPD, since the beam-splitter acting on the classical state does not create any entanglement, we can rewrite the effect of the auxiliary mode as a pre-factor on each term of the decomposition.

\subsection{E. Cat-state amplification}
Recalled that our Bell-like state is defined as
\begin{align}
     |\Phi(\alpha,\theta)\rangle = \frac{1}{\sqrt{2(1+\cos{\theta}e^{-4|\alpha|^2})}}(\ketv{\alpha,\alpha}+e^{i\theta}\ketv{-\alpha, -\alpha} )
\end{align}
Where the QPD of the state $|\Phi(\alpha,\theta)\rangle$ can be constructed form:
\begin{align}
    \proj{\Phi(\alpha,\theta)}\propto &(\proj{\alpha})^{\otimes 2}+(\proj{-\alpha})^{\otimes 2}
    +e^{-i\theta}(|\alpha\rangle\langle-\alpha|)^{\otimes 2}+e^{i\theta}(|-\alpha\rangle\langle\alpha|)^{\otimes 2} 
\end{align}
For the non-diagonal parts (last two terms), we first defined four cat-states as
\begin{align}
    |\psi_n\rangle = \frac{1}{\sqrt{\mathcal{N}_n}}(|\alpha\rangle+e^{i\frac{\theta+n\pi}{2}}|-\alpha\rangle),
\end{align}
with $n\in \mathbb{Z}_4$ and $\mathcal{N}_n = 2+2\cos(\frac{\theta+n\pi}{2})e^{-2|\alpha|^2}$.
Then we may further set
\begin{align}
    \hat{\sigma}_n:=\frac{\mathcal{N}_n}{2}\proj{ {\psi}_n} -\frac{\mathcal{N}_{n+2}}{2}\proj{ {\psi}_{n+2}},
\end{align}
and it follows that
\begin{align}
    \frac{1}{2}(\hat{\sigma}_0^{\otimes 2} - \hat{\sigma}_1^{\otimes2})= e^{-i\theta} (|\alpha \rangle\langle-\alpha|)^{\otimes 2}+e^{i\theta} (|-\alpha\rangle\langle\alpha|)^{\otimes 2} 
\end{align}
which is the desired non-diagonal term.
So the QPD of the state $\hat{\Phi}(\alpha,\theta)$ only contain $|\psi_n\rangle$ and $|\pm\alpha\rangle$.
The corresponding overhead is then
\begin{align}
    \bar{\gamma}(\alpha,\theta) &=\frac{1}{2(1+\cos{\theta}e^{-4|\alpha|^2})}(2+\frac{(\mathcal{N}_0+\mathcal{N}_2)^2+(\mathcal{N}_1+\mathcal{N}_3)^2}{8}) = \frac{3}{1+\cos{\theta}e^{-4|\alpha|^2}}.
\end{align}
\end{document}